\documentclass[12pt]{article}
\usepackage{amsmath,amssymb,amsthm}
\usepackage{lmodern}
\usepackage{fbb}
\usepackage[libertine]{newtxmath}
\usepackage[margin=2cm,a4paper]{geometry}
\usepackage{keytheorems}
\usepackage{zref-clever}
\usepackage{booktabs}
\usepackage{xcolor}
\usepackage{textpos}
\usepackage{graphicx}
\zcsetup{cap=true}
\usepackage{hyperref}
\usepackage{authblk,enumitem}
\hypersetup{
    colorlinks=true,
    citecolor=blue,
    linkcolor=blue,
    filecolor=blue,      
    urlcolor=blue,
    pdftitle={Branch-width of connectivity functions is fixed-parameter tractable},
    pdfauthor={Tuukka Korhonen and Sang-il Oum}
}
\newkeytheorem{theorem}[name=Theorem,parent=section,refname={Theorem,Theorems}, Refname={Theorem,Theorems}]
\newkeytheorem{lemma,proposition,corollary,conjecture,problem}[sibling=theorem]
\newcommand\abs[1]{\lvert #1\rvert}
\newcommand\bw{\operatorname{bw}}
\newcommand\SUB{\operatorname{SUB}}
\newcommand\co[1]{\overline{#1}}

\title{Branch-width of connectivity functions is fixed-parameter tractable}
\author{Tuukka Korhonen\thanks{Supported by the European Union under Marie Skłodowska-Curie Actions (MSCA), project no. 101206430, and by the VILLUM Foundation, Grant Number 54451, Basic Algorithms Research Copenhagen (BARC).
}}
\affil{University of Copenhagen, Copenhagen, Denmark}
\author{Sang-il~Oum\thanks{Supported by the Institute for Basic Science (IBS-R029-C1)}}
\affil{Discrete Mathematics Group, Institute for Basic Science (IBS), Daejeon,~South~Korea}
\affil{Department of Mathematical Sciences, KAIST, Daejeon, South~Korea}
\affil[ ]{\small \textit{Email addresses:} \texttt{tuko@di.ku.dk},
\texttt{sangil@ibs.re.kr}}
\date{January 8, 2026; revised February 6, 2026}

\begin{document}
\maketitle

 \begin{textblock}{20}(-0.5, 8.3)
 \includegraphics[width=100px]{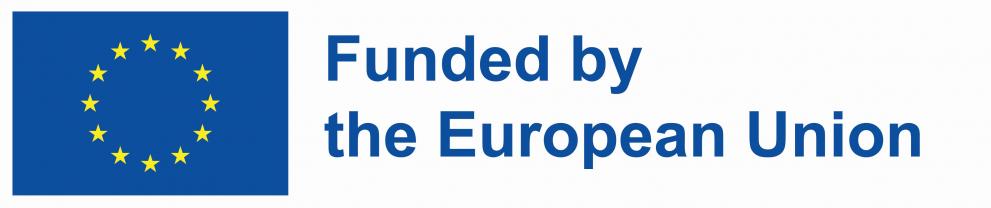}\end{textblock}

\begin{abstract}
    A connectivity function on a finite set $V$ is a symmetric submodular function $f \colon 2^V \to \mathbb{Z}$ with $f(\emptyset)=0$.
    We prove that finding a branch-decomposition of width at most $k$ for a connectivity function given by an oracle is fixed-parameter tractable (FPT), by providing an algorithm of running time $2^{O(k^2)} \gamma n^6 \log n$, where $\gamma$ is the time to compute $f(X)$ for any set $X$, and $n = |V|$.
    This improves the previous algorithm by Oum and Seymour [J. Combin. Theory Ser.~B, 2007], which runs in time $\gamma n^{O(k)}$.
    Our algorithm can be applied to rank-width of graphs, branch-width of matroids, branch-width of (hyper)graphs, and carving-width of 
    graphs.
    This resolves an open problem asked by Hlin\v{e}n\'y [SIAM J. Comput., 2005], who asked whether branch-width of matroids given by the rank oracle is fixed-parameter tractable.
    Furthermore, our algorithm improves the best known dependency on $k$ in the running times of FPT algorithms for graph branch-width, rank-width, and carving-width.
\end{abstract}

\section{Introduction}\label{sec:intro}
In 2005, Hlin\v{e}n\'y~\cite{Hlineny2002} asked the following questions for matroids.

\begin{quote}
    What is the parameterized complexity of the problem to determine the branch-width of a matroid $M$?
    \begin{enumerate}[label=\rm(\arabic*)]
        \item If $M=M(A)$ is given by a matrix representation over an infinite field?
        \item If $M$ is given by a rank oracle?
    \end{enumerate}
\end{quote}
We resolve this long-standing open problem completely, in a more general setting, by showing that the branch-width of connectivity functions is fixed-parameter tractable (FPT).
A \emph{connectivity function} is an integer-valued function $f$ defined on all subsets of a finite set $V$ such that 
\begin{enumerate}[label=\rm (\roman*)]
    \item (symmetric) $f(X)=f(V-X)$ for all sets $X$,
    \item (submodular) $f(X)+f(Y)\ge f(X\cap Y)+f(X\cup Y)$ for all sets~$X$ and $Y$, and 
    \item $f(\emptyset)=0$.
\end{enumerate}

Branch-width of graphs and matroids was introduced by Robertson and Seymour in their Graph Minors series~\cite{RS1991}.
It was later generalized for connectivity functions by Oum and Seymour~\cite{OS2004,OS2005}.
Let us provide the general definition of branch-width of connectivity functions.

A tree is \emph{subcubic} if every vertex has degree $3$ or $1$.
Let $f:2^V\to\mathbb Z$ be a connectivity function on a set $V$.
A \emph{branch-decomposition} of~$f$ is a pair $(T,L)$ of a subcubic tree $T$ and a bijection $L$ from the set of leaves of $T$ to $V$.
For a branch-decomposition $(T,L)$ of~$f$ and an edge~$e$ of~$T$, 
we define the \emph{width} of~$e$ as $f(L^{-1}(A_e))$, 
where $(A_e,B_e)$ is the partition of $V(T)$ induced by the components of $T-e$.
The \emph{width} of a branch-decomposition $(T,L)$ is the maximum width of all edges of $T$.
The \emph{branch-width} of~$f$, denoted by $\bw(f)$, is the minimum width of all branch-decompositions of~$f$.
If $\abs{V}\le 1$, there is no branch-decomposition and yet we define the branch-width of $f$ to be $0$.

Hicks and Oum~\cite{HO2011} wrote the following remark in 2011.
\begin{quote}
    For many applications on fixed-parameter tractable algorithms, it is desirable to have an algorithm which runs in time $O(g(k)n^c)$ for some function $g$ and a constant $c$ independent of $k$. Such
an algorithm is called a \emph{fixed-parameter tractable} algorithm with parameter $k$. It is still unknown
whether there is a fixed-parameter tractable algorithm to decide whether branch-width of $f$ is at
most $k$ when $f$ is an integer-valued symmetric submodular function given as an oracle.
\end{quote}

Here is our main theorem, proving that there is a fixed-parameter tractable algorithm to find a branch-decomposition of width at most $k$ if one exists for general connectivity functions.
\begin{theorem}[label=thm:main-simplified]
    Let $n>1$ be an integer and $f:2^V\to \mathbb Z$ a connectivity function on an $n$-element set~$V$.
    Let $\gamma$ be the time to compute $f(X)$ for any subset $X$ of~$V$.
In time
    $2^{O(k)}\gamma  n^6 \log n + 2^{O(k^2)} \gamma n$,
we can either find a branch-decomposition of~$f$ of width at most~$k$, 
    or confirm that the branch-width of~$f$ is larger than~$k$.
\end{theorem}

The branch-width of a matroid $M$ on a ground set $E(M)$ is defined as the branch-width of the connectivity function $\lambda(X) = r(X) + r(E(M)-X) - r(E(M))$ of the matroid, where $r(X)$ denotes the rank of a set $X \subseteq E(M)$.
Therefore, \zcref{thm:main-simplified} answers both of the questions of Hlin\v{e}n\'y affirmatively.
Furthermore, we can improve its running time by using algorithms dedicated for matroids and obtain the following.

\begin{theorem}[label=thm:matroidmain,store=matroidmain]
    There is an algorithm that, with input an $n$-element matroid $M$, given by its rank oracle, and an integer $k$, finds a branch-decomposition of width at most $k$, if one exists, in time $2^{O(k)} \gamma n^{2.5}\log^2 n +  2^{O(k^2)}\gamma n$, where $\gamma$ is the time to compute the rank of any set.     
\end{theorem}

The previous best algorithm for branch-width of a general connectivity function was given by Oum and Seymour~\cite{OS2005} in 2007.
They showed that the problem of computing branch-width of a connectivity function given by an oracle is slice-wise polynomial (XP) by giving an algorithm with running time $\gamma n^{O(k)}$.
Slightly earlier, Oum and Seymour~\cite{OS2004} gave an FPT approximation algorithm for branch-width of connectivity functions.
Their algorithm runs in time $2^{O(k)} \gamma n^{O(1)} $, and either determines that the branch-width is more than $k$, or returns a branch-decomposition of width at most $3k+c$, where $c \le k$ is the maximum of $f(\{v\})$ for $v \in V$.
Oum~\cite{Oum2009} also presented an exponential-time algorithm to compute branch-width.

For branch-width of matroids represented over a fixed finite field, Hlin\v{e}n\'y and Oum~\cite{HO2006} gave an FPT algorithm in~2008 (see also~\cite{Hlineny2002,JKO2019}).
For more general settings of matroid branch-width, the problem of finding an FPT algorithm remained open before this work.

In addition to branch-width of matroids, the branch-width of connectivity functions captures several well-studied graph width parameters, including rank-width, carving-width, and branch-width of graphs.
Each of these was known to be fixed-parameter tractable: Bodlaender and Thilikos~\cite{BT1997} gave a linear FPT algorithm for branch-width of graphs, Thilikos, Serna, and Bodlaender~\cite{TSB2000} gave a linear FPT algorithm for carving-width, and Courcelle and Oum~\cite{DBLP:journals/jct/CourcelleO07} gave an FPT algorithm for rank-width.
The FPT algorithm for rank-width was later improved to almost-linear by Korhonen and Soko{\l}owski~\cite{KS2024} (see also~\cite{HO2006,Oum2006,JKO2019,FK2024} for other works on computing rank-width).

Despite the earlier FPT algorithms for the aforementioned special cases of branch-width of connectivity functions, we believe that the algorithm of \zcref{thm:main-simplified} is interesting also in their context.
Its running time dependency on the parameter $k$ is only $2^{O(k^2)}$, while for the algorithms for graph branch-width and carving-width~\cite{BT1997,TSB2000} it is at least $2^{\Omega(k^3)}$, and for the algorithms for rank-width and finite field matroid branch-width~\cite{DBLP:journals/jct/CourcelleO07,HO2006,JKO2019,FK2024,KS2024} it is at least doubly exponential.
Furthermore, in our opinion, the algorithm of \zcref{thm:main-simplified} provides so far the simplest proof of fixed-parameter tractability for all of these width parameters.
It is an interesting direction of future research to give faster problem-specific implementations of the technique of \zcref{thm:main-simplified} for these special cases.

\paragraph{Outline of the algorithm.}
The proof of \zcref{thm:main-simplified} is not long, but let us still shortly outline it.
We say that a \emph{cut} of a connectivity function $f \colon 2^{V} \to \mathbb{Z}$ is a bipartition $(A,B)$ of $V$.
A \emph{safe cut} is a cut $(A,B)$ such that there exists an optimum-width branch-decomposition where $(A,B)$ corresponds to an edge of the decomposition.

The main idea of the algorithm of \zcref{thm:main-simplified} is to repeatedly find safe cuts and greedily decompose $f$ along them.
In particular, if we find a safe cut $(A,B)$ with $|A|,|B| \ge 2$, we can construct two smaller connectivity functions $f_A$ and $f_B$, recursively find optimum-width branch-decompositions of them, and combine them to an optimum-width branch-decomposition of $f$.

The surprising fact that makes the above idea work is that we can always find safe cuts, provided that the branch-width of $f$ is small enough compared to $|V|$.
In particular, we show that if we are given a branch-decomposition of $f$ of width $\ell$, where $3^{\ell+1} < |V|$, then we can find a safe cut $(A,B)$ with $|A|,|B| \ge 2$ in time $2^{O(\ell)}\gamma n^{O(1)} $.
Therefore, by using either iterative compression or the FPT approximation algorithm of Oum and Seymour~\cite{OS2004}, we can decompose along safe cuts until $|V| \le 2^{O(k)}$.
Once $|V| \le 2^{O(k)}$, we can use the XP algorithm of Oum and Seymour~\cite{OS2005} to find an optimum-width branch-decomposition in $2^{O(k^2)} \gamma$ time.

\paragraph{Organization of the paper.}
The paper is organized as follows.
In \zcref{sec:prelim} we review the necessary definitions and preliminary results.
In \zcref{sec:search} we provide an algorithm for testing if a cut satisfies a certain condition that implies it is safe. 
In \zcref{sec:titanic} we provide an algorithm for turning cuts into safe cuts.
In \zcref{sec:algorithm} we give our algorithm, and in \zcref{sec:matroid} the optimized version for matroids.
We conclude with discussions and open questions in \zcref{sec:conclusion}.

\section{Preliminaries}\label{sec:prelim}
We start this section by reviewing several properties of submodular functions and connectivity functions.
A function $f:2^V\to\mathbb Z$ is \emph{submodular} if $f(X)+f(Y)\ge f(X\cap Y)+f(X\cup Y)$ for all $X,Y\subseteq V$.
For $X \subseteq V$, we use the notation $\co{X} = V-X$.
The function $f$ is \emph{symmetric} if $f(X) = f(\co{X})$ for all $X \subseteq V$, and a \emph{connectivity function} if it is submodular, symmetric, and $f(\emptyset) = 0$.

We remark that the condition $f(\emptyset) = 0$ is natural, because for every symmetric and submodular function $f$, the function $f'(X) = f(X) - f(\emptyset)$ is a connectivity function with non-negative image.

\begin{theorem}[label=thm:os,note={Oum and Seymour~\cite{OS2005}}]
   Let $n$, $k$ be positive integers.
   Let $f:2^V\to \mathbb Z$ be a connectivity function on an $n$-element set~$V$.
   Let $\gamma$ be the time to compute $f(X)$ for any subset $X$ of~$V$.
   There is an algorithm to 
   find a branch-decomposition of~$f$ of width at most $k$ if one exists, 
   in time $O(\gamma n^{8k+9}\log n)$.
\end{theorem}

We say that a function $g:2^V\to\mathbb Z$ is \emph{nondecreasing} if for all subsets $X\subseteq Y\subseteq V$, we have $g(X)\le g(Y)$.
A \emph{bipartition} of a set $V$ is a pair $(A,B)$ of disjoint sets with $A \cup B = V$, and a \emph{tripartition} of $V$ is a triple $(A,B,C)$ of disjoint sets with $A \cup B \cup C = V$.
We allow the sets $A,B,C$ to be empty.

We recall a property of connectivity functions that is sometimes called the \emph{posimodularity}.

\begin{lemma}[label=lem:diff]
    Let $f:2^V\to \mathbb Z$ be a connectivity function on a finite set~$V$.
    Then for subsets $X$ and $Y$ of $V$, 
    \[ f(X)+f(Y)\ge f(X-Y)+f(Y-X).    \] 
\end{lemma}
\begin{proof}
    $f(X)+f(Y)=f(X)+f(\co{Y})\ge f(X-Y)+f(X\cup \co{Y})= f(X-Y)+f(Y-X)$.
\end{proof}

As part of our algorithm we use the following submodular function minimization algorithm.

\begin{theorem}[note={Chakrabarty, Lee, Sidford, and Wong~\cite{CLSW2017}},label={thm:submodularmin}]
    Let $f:2^V\to\mathbb Z$ be a submodular function on an $n$-element set~$V$ such that $\abs{f(X)}\le M$ for all $X\subseteq V$.
    Let us assume that $\gamma$ is the time to compute $f(X)$ for any subset $X$ of~$V$.
    Then we can find $A \subseteq V$ minimizing $f(A)$ in time $O(\gamma n M^3 \log n)$. 
\end{theorem}

Let us write $\SUB(n,M,\gamma)$ to denote the running time to find a set $A$ that minimizes $f(A)$ for a submodular function~$f$ on an $n$-element set with $\abs{f(X)}\le M$ for all sets~$X$, 
where computing $f(X)$ takes the time $\gamma$ for any set $X$.
By \zcref{thm:submodularmin}, $\SUB(n,M,\gamma) \le O(\gamma n M^3 \log n)$.

For a proper nonempty subset~$A$ of~$V$, we define $f\lhd A$ to be the connectivity function on $V\lhd A = \co{A} \cup \{A\}$,
where  for $X\subseteq V\lhd A$, we have 
\[ 
    (f\lhd A)(X)= 
    \begin{cases}
        f(X) & \text{if } A\notin X,\\ 
        f((X-\{A\})\cup A) &\text{if }A\in X.
    \end{cases}
\] 
In other words, in $V\lhd A$, the elements of $A$ are merged into one element $A$
and $f\lhd A$ is a connectivity function obtained from $f$ by merging elements of $A$ into one.

It is easy to see that if both $f\lhd A$ and $f\lhd \co{A}$ have branch-width at most $k$, then so does $f$.
In particular, branch-decompositions $(T_A,L_A)$ of $f \lhd A$ and $(T_{\co{A}}, L_{\co{A}})$ of $f \lhd \co{A}$ can be merged into a branch-decomposition of $f$ by combining them from the leaves corresponding to $A$ and $\co{A}$, and the width of the resulting decomposition is the maximum of the widths of $(T_A,L_A)$ and $(T_{\co{A}}, L_{\co{A}})$.

The key ingredient of this paper is a sufficient condition for the converse, which uses a property defined by Robertson and Seymour~\cite{RS1991}, called \emph{titanic}.
A proper nonempty subset $A$ of~$V$ is \emph{titanic} if for every tripartition $(A_1,A_2,A_3)$ of $A$, there exists $i\in\{1,2,3\}$ such that $f(A_i)\ge f(A)$.
Note that we allow $A_i$ to be empty.

The following lemma is from Geelen, Gerards, Robertson, and Whittle~\cite[Lemma 2.1]{GGRW2003a}. They wrote it for the matroid connectivity functions but their proof works completely for general connectivity functions. 
We strengthen their lemma, by allowing $\co{A}$ to be titanic instead of requiring $f(A)\le k$.
In particular, the sufficient condition is that both $A$ and $\co{A}$ are titanic.
A similar lemma was also recently shown in the work of Korhonen~\cite[Lemmas~6.22~and~6.27]{DBLP:conf/stoc/Korhonen25}.

\begin{proposition}[label=prop:branch]
    Let $f:2^V\to\mathbb Z$ be a connectivity function on a finite set $V$. 
    Let $(T,L)$ be a branch-decomposition of~$f$ of width at most $k$.
    Let $A$ be a titanic set.
    If $\co{A}$ is titanic or $f(A)\le k$, then we can convert $(T,L)$ into a branch-decomposition of $f\lhd A$ of width at most $k$ in time $O(\gamma \abs{V})$, where $\gamma$ is the time to compute $f(X)$ for any set $X$.
\end{proposition}
\begin{proof}
    We may assume that $\abs{A}>1$ and $A\neq V$.
    For an edge $e$ of $T$ and a node $u$ of $T$, let $X_{e,u}$ be the set of elements of $V$ that are mapped by $L$ to nodes in the component of $T-e$ not containing $u$.
    For every edge $e=uv$ of $T$, 
    let us consider the partition $(X_{e,u}\cap A,X_{e,v}\cap A)$ of $A$. 
    Since $A$ is titanic, $f(X_{e,u} \cap A)\ge f(A)$
    or $f(X_{e,v}\cap A)\ge f(A)$.
    By submodularity, if $f(X_{e,v}\cap A)\ge f(A)$, then 
    \[  
    f(X_{e,v})+f(A)\ge f(X_{e,v}\cap A)+f(X_{e,v}\cup A) \ge f(A)+f(X_{e,u}-A), 
    \] 
    and therefore $f(X_{e,u}-A)\le f(X_{e,v})=f(X_{e,u})$.
    Similarly, if $f(X_{e,u}\cap A)\ge f(A)$, then $f(X_{e,v} - A)\le f(X_{e,v})$. 
    Thus, $f(X_{e,u}-A)\le f(X_{e,u})$ or $f(X_{e,v}-A)\le f(X_{e,v})$.
    Let us orient an edge $e=uv$ towards~$v$ if $f(X_{e,v}-A)\le f(X_{e,v})$.
    If both $f(X_{e,u}-A)\le f(X_{e,u})$ and $f(X_{e,v}-A)\le f(X_{e,v})$, then we orient the edge $e$ in both ways. 
    However, if $e$ is incident with a leaf, then we assume that $e$ is not oriented towards the leaf. This assumption is fine, because $e$ is oriented away from the leaf as 
    for every $v\in V$, $f(\{v\}-A)\le f(\{v\})$.
    Thus every edge of $T$ receives at least one orientation.

    Suppose that there are two edges $e$ and $d$ of~$T$ such that $e$ is not oriented towards an end of $d$ and $d$ is not oriented towards an end of $e$.
    Let $w$ be a degree-$3$ node of $T$ in the component of $T-\{e,d\}$ containing both an end of $e$ and an end of $d$.
    Let $Y_1=X_{ew}$, $Y_3=X_{dw}$, and $Y_2=V-(Y_1\cup Y_3)$.
    Since neither $e$ nor $d$ is oriented towards $w$, we have 
    $f(\co{Y_1}-A)\le f(\co{Y_1})$
    and $f(\co{Y_3}-A)\le f(\co{Y_3})$.
    Then
    \begin{align*}
    f(Y_1)+f(Y_3) &= f(\co{Y_1}) + f(\co{Y_3})\\
    &\ge f(\co{Y_1}-A) + f(\co{Y_3}-A)\\
    &\ge f(Y_1 \cup A) + f(\co{Y_3}-A) && \text{(symmetry)}\\
    &\ge f(A \cup Y_1 \cup Y_2) + f(Y_1 - A) && \text{(submodularity)}\\
    &\ge f(Y_3 - A) + f(Y_1 - A), && \text{(symmetry)}
    \end{align*}
    and therefore $f(Y_1-A)\le f(Y_1)$ or $f(Y_3-A)\le f(Y_3)$, contradicting the assumption that neither~$e$ nor~$d$ is oriented towards $w$.

    Let $T'$ be a directed tree obtained from $T$ by contracting all bidirected edges.
    Then every edge of $T'$ has exactly one orientation.
    Since $\abs{E(T')}<\abs{V(T')}$, there is a node $x$ of $T'$ such that every edge incident with $x$ is oriented towards~$x$.
    The node $x$ of $T'$ corresponds to a set $X$ of nodes in $T$
    and we pick any $s$ in $X$.
    Then 
    every node of $T$ has a directed path to $s$.
    Thus, $f(X_{es}-A)\le f(X_{es})$ for all edges $e$ of~$T$.

    Let $e_1$, $e_2$, $e_3$ be the three edges of~$T$ incident with $s$.
    Let $X_i=X_{e_is}$ for all $i\in\{1,2,3\}$.

    If $\co{A}$ is titanic, then there is $i\in \{1,2,3\}$ such that $f(X_i-A)\ge f(\co{A})$ and so $f(\co{A})\le f(X_i)\le k$.
    So we may now assume that $f(A)\le k$.

    Since $A$ is titanic, there is $i\in \{1,2,3\}$ such that $f(X_i\cap A) \ge f(A)$. Without loss of generality, let us assume that $f(X_1\cap A)\ge f(A)$.
    By submodularity, 
    \[ f(X_1)+f(A)\ge f(X_1\cap A)+f(X_1\cup A)\ge f(A)+f(\co{X_1}-A), \] 
    and therefore $f(\co{X_1}-A)\le f(X_1)$.
    Now to obtain a new branch-decomposition $(T_A,L_A)$ of $f \lhd A$ of width at most $k$,
    from $(T,L)$, 
    we subdivide the edge $e_1$ to insert a node $w$
    and attach a leaf $w'$ adjacent to $w$
    such that $L_A$ maps $A$ to $w'$. 
    We remove all existing leaves of $T$ that were images of a vertex in $A$ by $L$ and smoothen all degree-$2$ vertices. Now it is easy to verify that the width of $(T_A,L_A)$ is at most $k$.

    The branch-decomposition $(T_A,L_A)$ can be constructed in time $O(\gamma \abs{V})$ by evaluating $f(X_{e,u})$, $f(X_{e,u}-A)$, and $f(X_{e,v}-A)$ for every edge $e = uv$ of $T$.
\end{proof}

\zcref{prop:branch} implies that if both $A$ and $\co{A}$ are titanic, then there is an optimum-width branch-decomposition $(T,L)$ of $f$, where an edge of $T$ corresponds to the bipartition $(A,\co{A})$ of~$V$.
Hence also $\bw(f) \ge f(A)$.

\section{Testing if a set is titanic}\label{sec:search}
As a key subroutine of our algorithm, we need an algorithm for testing if a given set $A \subseteq V$ is titanic, and if not, finding a tripartition of $A$ witnessing that.
In this section, we give such an algorithm by using an observation that it corresponds to finding whether a polymatroid can be covered by three proper flats.

A polymatroid $g$ on a finite set $V$ is a function $g:2^V\to\mathbb{Z}$ 
such that
\begin{enumerate}[label=\rm(\roman*)]
    \item (nondecreasing) if $X\subseteq Y$, then $g(X)\le g(Y)$, 
    \item (submodular) for all $X$ and $Y$, we  have $g(X)+g(Y)\ge g(X\cap Y)+g(X\cup Y)$, and
    \item $g(\emptyset)=0$.
\end{enumerate}
For a set $X$, we call $g(X)$ the \emph{rank} of $X$.
A set $X$ is a \emph{flat} if it is a maximal set containing $X$ of rank $g(X)$. 
In other words, $X$ is a flat if and only if $g(X\cup\{a\})>g(X)$ for all $a\in \co{X}$.
A set $X$ is \emph{proper} if $X\neq V$.

Next, we give an algorithm for testing if a polymatroid can be covered by three proper flats.

\begin{proposition}[label=prop:coverfpt]
    Let $g:2^A\to\mathbb Z$ be a polymatroid and $\gamma$ be the time to compute $g(U)$ for any subset~$U$ of~$A$.
    Then in time $O(\gamma 3^{3g(A)} \abs{A})$, we can either
    \begin{itemize}
        \item find 
        a triple $(X,Y,Z)$ of three subsets of $A$ such that 
        $\max(g(X),g(Y),g(Z))<g(A)$ and $X \cup Y \cup Z = A$, or
        \item confirm that no such triple $(X,Y,Z)$ exists.
    \end{itemize}
\end{proposition}
\begin{proof}
    We design a recursive algorithm that takes as input additionally three sets $X',Y',Z'$, and solves the problem with the additional constraints that $X' \subseteq X$, $Y' \subseteq Y$, and $Z' \subseteq Z$.
    The original problem corresponds to this with $X'=Y'=Z'=\emptyset$.

We proceed by induction on $3g(A)-(g(X')+g(Y')+g(Z'))$.
    By submodularity, we may assume that $X'$, $Y'$, and $Z'$ are flats, by replacing each with a maximal set containing it with the same rank.
    Such replacement can be done in time $O(\gamma \abs{A})$.
    If $X' \cup Y' \cup Z'=A$, then we found a desired triple. 
    If $X'\cup Y'\cup Z' \neq A$, then pick an arbitrary $v\in A-(X'\cup Y'\cup Z')$.
    Then we run the recursive algorithm for three instances $(X'\cup\{v\},Y',Z')$, $(X',Y'\cup\{v\},Z')$, and $(X',Y',Z'\cup\{v\})$.
    Note that $g(X'\cup\{v\})>g(X')$, $g(Y'\cup\{v\})>g(Y')$, and $g(Z'\cup\{v\})>g(Z')$.

    Therefore, the total running time satisfies the recursion $T(\abs{A}, 3g(A)-(g(X')+g(Y')+g(Z'))) \le O(\gamma \abs{A})+ 3 T(\abs{A},3g(A)-(g(X')+g(Y')+g(Z'))-1)$.
    Thus $T(\abs{A},3g(A)-(g(X')+g(Y')+g(Z'))) \le O(\gamma 3^{3g(A)-(g(X')+g(Y')+g(Z'))} \abs{A})$, and $T(\abs{A}, 3g(A)-3g(\emptyset)) \le O(\gamma 3^{3g(A)} \abs{A})$.
\end{proof}

The following lemma is well known; for instance, it appeared in \cite[Lemma 3.2]{HO2006}.
\begin{lemma}[label=lem:makepartition]
    Let $f:2^V\to \mathbb Z$ be a connectivity function on a set~$V$ and $A$ a subset of~$V$.
    If $C_1$, $C_2$, $C_3$ are subsets of $A$ such that $C_1\cup C_2\cup C_3=A$ and $f(C_i)<f(A)$ for all $i\in\{1,2,3\}$, then 
    we can find a tripartition $(C_1',C_2',C_3')$ of $A$ such that 
    $f(C_i')<f(A)$ for all $i\in\{1,2,3\}$
    by computing~$f$ at most $3$ times.
\end{lemma}

\begin{proof}
    If $C_1\cap C_2$ is nonempty, then 
    by \zcref{lem:diff}, we have $f(C_1-C_2)+f(C_2-C_1)<2f(A)$ and therefore 
    $f(C_1-C_2)<f(A)$ or $f(C_2-C_1)<f(A)$. If $f(C_1-C_2)<f(A)$, then we replace~$C_1$ by $C_1-C_2$ and otherwise replace $C_2$ by $C_2-C_1$.
    Thus, by computing~$f$ once, we can make $C_1\cap C_2=\emptyset$.
    We do the same step for pairs $(C_2,C_3)$ and $(C_3,C_1)$. 
\end{proof}

We now complete the main goal of this section.

\begin{proposition}[label=prop:titanic]
    Let $f:2^V\to \mathbb Z$ be a connectivity function on a set~$V$ and $\gamma$ the time to compute $f(X)$ for any subset $X$ of~$V$.
    For a given subset $A \subseteq V$, we can
    in time $O(3^{3f(A)} \abs{A} \SUB(\abs{A},\max_{B\subseteq A} f(B),\gamma) )$ 
    either find a tripartition $(C_1,C_2,C_3)$ of $A$ such that $f(C_i)<f(A)$ for all $i\in\{1,2,3\}$, or confirm that there is no such partition, that is, $A$ is titanic.
\end{proposition}
\begin{proof}
    Let $g(X)=\min_{X\subseteq Z\subseteq A} f(Z)$ for $X\subseteq A$.
    Then $g \colon 2^A \to \mathbb{Z}$ is a polymatroid.
    Note that computing $g(X)$ takes time 
    $\SUB(\abs{A},\max_{B\subseteq A} f(B),\gamma)$.
By \zcref{lem:makepartition}, it is enough to find three subsets $X$, $Y$, $Z$ of~$A$ such that $X\cup Y\cup Z=A$ and $\max(f(X),f(Y),f(Z))<f(A)$
    or confirm that such three subsets do not exist.    
    By applying the algorithm of \zcref{prop:coverfpt}, we either find a triple $(X,Y,Z)$ of subsets of $A$ such that $\max(g(X),g(Y),g(Z))<g(A)$ or confirm that no such triple $(X,Y,Z)$ exists,
    in time $O(3^{3g(A)} \abs{A} \SUB(\abs{A},\max_{B\subseteq A} f(B),\gamma))$.
    If we obtain a triple $(X,Y,Z)$, then by making each of $X$, $Y$, $Z$ maximal while keeping the same $g$ value, 
    we ensure that $g(X)=f(X)$, $g(Y)=f(Y)$, and $g(Z)=f(Z)$.
\end{proof}

\section{Partitioning into two titanic sets}\label{sec:titanic}
We give a proposition that allows us to find a set $A \subseteq V$ so that both $A$ and $\co{A}$ are titanic.
The idea for this proposition originates from the recent work of Korhonen~\cite[Lemma~9.6]{DBLP:conf/stoc/Korhonen25}, although there it is stated completely differently.

\begin{proposition}[label=prop:split]
    Let $f:2^V\to\mathbb Z$ be a connectivity function of branch-width at most $k$ and denote $n = |V|$.
There exists a subset $A$ of~$V$ satisfying the following two conditions.
    \begin{enumerate}[label=\rm(\roman*)]
        \item $\abs{A} \ge \frac{n}{3^{k+1}}$ and $\abs{\co{A}} \ge \frac{n}{3^{k+1}}$.
        \item Both $A$ and $\co{A}$ are titanic.
    \end{enumerate}
    Moreover, if a branch-decomposition of width at most $k$ is given,
    then we can find such a set $A$ in time $O(
    3^{3k}k n \SUB(n,kn,\gamma))$, where $\gamma$ is the time to compute $f(X)$ for any set $X$.
\end{proposition}

\begin{proof}
    Let $(T,L)$ be a branch-decomposition of $f$ of width at most $k$. 
    Then we can find a partition $(A_0,B_0)$ of $V$ such that 
    $\abs{A_0} \ge n/3$, $\abs{B_0} \ge n/3$, and $f(A_0)\le k$.

    Suppose that we have a partition $(A_i,B_i)$ of $V$ such that 
$\abs{A_i} \ge \frac{n}{3^{i+1}}$, $\abs{B_i} \ge \frac{n}{3^{i+1}}$, and $f(A_i)\le k-i$
    for some integer $i\ge0$. 
    By applying the algorithm of \zcref{prop:titanic} to both~$A_i$ and~$B_i$, we can decide if $A_i$ or $B_i$ is titanic. 
    If both $A_i$ and $B_i$ are titanic, then we found a desired set $A:=A_i$. 
    Thus we may assume that at least one of $A_i$ and $B_i$ is not titanic. 
    Without loss of generality, let us assume that $A_i$ is not titanic. 
    The algorithm in \zcref{prop:titanic} provides a partition $(C_1,C_2,C_3)$ of $A_i$ such that $f(C_1), f(C_2), f(C_3)<f(A_i)$.
    Without loss of generality, let us assume that $\abs{C_1}\ge \abs{A_i}/3$.
    Let $A_{i+1} = C_1$ and let $B_{i+1}=V-A_{i+1}$. 
    Then $\abs{A_{i+1}}\ge \frac13 \abs{A_i} \ge \frac{n}{3^{i+2}}$. 
    Trivially, $\abs{B_{i+1}}\ge \abs{B_i} \ge \frac{n}{3^{i+1}} \ge \frac{n}{3^{i+2}}$.
    Since $f(A_{i+1})<f(A_i)$, we deduce that $f(A_{i+1})\le k-(i+1)$.

    Since $f(A_{i+1})\ge 0$, this process must stop at some $0\le i<k$.
    Thus we have found a desired set $A$.

    This algorithm can be implemented by $O(k)$ calls to the algorithm of \zcref{prop:titanic} and the oracle for $f$.
    Because $\bw(f) \le k$, we have $f(\{v\}) \le k$ for all $v \in V$, so the submodularity of $f$ implies $\max_{X \subseteq V} f(X) \le kn$.
    Thus, the algorithm runs in time 
    $O(k\cdot (\gamma+3^{3k} n \SUB(n,kn,\gamma)))$.
\end{proof}

\section{Algorithm}
\label{sec:algorithm}
We then present our main algorithm.

\begin{proposition}[label=prop:approx-exact]
    Let $n$ be a positive integer.
    Let $f:2^V\to \mathbb Z$ be a connectivity function on an $n$-element set~$V$.
    Let us assume that $\gamma$ is the time to compute $f(X)$ for any subset $X$ of~$V$.
    Suppose that we are given a branch-decomposition of~$f$ of width at most $\ell$.
    In time 
$O(3^{4\ell}\ell n \SUB(n,k n,\gamma)\log n
    + \ell 3^{(\ell+1)(8k+9)}\gamma n)$, 
    we can either find a branch-decomposition of~$f$ of width at most~$k$, 
    or confirm that the branch-width of~$f$ is larger than~$k$.    
\end{proposition}
\begin{proof}
    We proceed by induction on~$n$, and prove that there is such an algorithm with the running time $T(n)$.
    We may assume that $\ell>k$.
    We may also assume that $f(\{v\})\le k$ for all $v\in V$, because otherwise the branch-width of $f$ is larger than $k$.
    Therefore $f(X)\le kn$ for all sets~$X$ by submodularity.
    
    If $n\le 3^{\ell+1}$, then we use \zcref{thm:os} to find a branch-decomposition of width at most $k$ in time $O(\gamma (3^{\ell+1})^{8k+9}\log 3^{\ell+1})\le O(\gamma \ell 3^{(\ell+1)(8k+9)})$.
    Thus we may assume that $n>3^{\ell+1}$.

    Let $(T,L)$ be a given branch-decomposition of width at most $\ell$.
    By \zcref{prop:split}, we can find a subset $A$  of~$V$ such that 
    $\abs{A} \ge 3^{-\ell-1}n > 1$, $\abs{\co{A}} \ge 3^{-\ell-1}n > 1$,
    and both $A$ and $\co{A}$ are titanic in time 
$O(3^{3\ell}\ell n \SUB(n,kn,\gamma))$.
    We may assume that $f(A)\le k$ because otherwise $f$ has branch-width larger than $k$ by \zcref{prop:branch}.

    Since $A$ is titanic, in time $O(\gamma n)$, we can convert $(T,L)$
    into a branch-decomposition $(T_1,L_1)$ of $f\lhd A$ of width at most $\ell$ by \zcref{prop:branch}.
    By running our algorithm for $f\lhd A$ given with $(T_1,L_1)$, 
    we either find a branch-decomposition $(T_1',L_1')$ of~$f\lhd A$ of width at most $k$ or confirm that the branch-width of $f\lhd A$ is larger than $k$, 
    which implies that the branch-width of $f$ is larger than $k$ by \zcref{prop:branch}.
    So we may assume that we obtain a branch-decomposition $(T_1',L_1')$ of~$f\lhd A$ of width at most $k$ in time $O(\gamma n)+T(\abs{\co{A}}+1)$.
    Similarly, since $\co{A}$ is titanic, we may assume that we obtain a branch-decomposition $(T_2',L_2')$ of~$f\lhd \co{A}$ of width at most $k$ in time $O(\gamma n)+T(\abs{A}+1)$.
    Now we glue $T_1'$ and $T_2'$ by attaching the leaf of $A$ in $T_1'$ with the leaf of $\co{A}$ in $T_2'$ 
    to obtain a branch-decomposition $(T',L')$ of $f$ of width at most $k$.

    For the running time, we have 
    \[ T(n)\le 
    O(3^{3\ell}\ell n \SUB(n,kn,\gamma))
    +O(\gamma n)+T(\abs{A}+1)+T(\abs{\co{A}}+1).\]
    Since $\max(\abs{A}+1,\abs{\co{A}}+1)\le (1-3^{-\ell-1})n+1 < n$, this recursion terminates and its depth is at most $O(3^{\ell}\log n)$.
    So we deduce that $T(n)\le O(3^{4\ell}\ell n \SUB(n,kn,\gamma)\log n
+  \ell 3^{(\ell+1)(8k+9)} \gamma n)$.
\end{proof}

To lift the algorithm of \zcref{prop:approx-exact} to the full algorithm of \zcref{thm:main-simplified}, we need to remove the assumption that we already have a branch-decomposition.
This can be done in two ways: by iterative compression and by using the FPT-approximation algorithm of Oum and Seymour~\cite{OS2004}, both resulting in roughly the same running time.
We present both of them here.

Let $3^V$ be the set of all pairs $(X,Y)$ of disjoint subsets of~$V$.
For a connectivity function $f:2^V\to\mathbb Z$, a function $f^*:3^V\to \mathbb Z$ is called an \emph{interpolation} of $f$, if it satisfies the following \cite{OS2004}.
\begin{enumerate}[label=\rm(\roman*)]
    \item $f^*(X,\co{X})=f(X)$ for all $X\subseteq V$.
    \item (monotone) If $C\cap D=\emptyset$, $A\subseteq C$, and $B\subseteq D$, then $f^*(A,B)\le f^*(C,D)$.
    \item (submodular) $f^*(A,B)+f^*(C,D)\ge f^*(A\cap C,B\cup D)+f^*(A\cup C, B\cap D)$.
    \item $f^*(\emptyset,\emptyset)=f(\emptyset)$.
\end{enumerate}
It is known that $f_{\min}(X,Y)=\min_{X\subseteq Z\subseteq V-Y} f(Z)$ is an interpolation of $f$ \cite[Proposition 4.2]{OS2004} and it can be computed by \zcref{thm:submodularmin}.
But in many cases, there is an interpolation that can be computed faster than $f_{\min}$, see \cite{OS2004}.

\begin{theorem}[label=thm:main,store=main]
    Let $n$ be a positive integer.
    Let $f:2^V\to \mathbb Z$ be a connectivity function on an $n$-element set~$V$.
Let $f^*:3^V\to\mathbb Z$ be an interpolation of $f$.
    Let $\gamma'$ be the time to compute $f^*(X,Y)$ for any disjoint subsets $X$, $Y$ of~$V$.
    In time 
$O(3^{8k}k n^2 \SUB(n,k n,\gamma')\log n
    +  k 3^{(2k+1)(8k+9)}\gamma'n^2)$
    we can either find a branch-decomposition of~$f$ of width at most~$k$, 
    or confirm that the branch-width of~$f$ is larger than~$k$.

    If $f^*=f_{\min}$ and computing $f(X)$ takes the time $\gamma$ for any set $X$, then it can be done in time $O(3^{8k} k n^2 \SUB(n,kn,\gamma)\log n+k 3^{(2k+1)(8k+9)} n^2 \SUB(n,kn,\gamma))$.
    In particular, in this case, 
    the running time is at most 
$2^{O(k^2)} \gamma n^6 \log^2 n$.
\end{theorem}
\begin{proof}
    We use the iterative compression technique.
    We may assume that $f(\{v_i\})\le k$ for each~$i$, because otherwise the branch-width of $f$ is larger than~$k$.
Let $\{v_1,v_2,\ldots,v_n\}=V$.
    Let $V_i=\{v_1,v_2,\ldots,v_i\}$ for each $i\in\{1,2,\ldots,n\}$.
Let $f_i(X)=f^*(X,V_i-X)$ for $X\subseteq V_i$.
    Then $f_i$ is a connectivity function on $V_i$
    and computing $f_i(X)$ can be done in time 
    $\gamma'$.

    If $(T,L)$ is a branch-decomposition of~$f$ of width at most~$k$,
    then we can restrict $T$ to the leaves in~$V_i$ to obtain a branch-decomposition of~$f_i$ of width at most $k$.
    Therefore, if $f$ has branch-width at most $k$, then every $f_i$ has branch-width at most $k$.

    Suppose that we have a branch-decomposition $(T_i,L_i)$ of $f_i$ of width at most $k$ for some $i\ge 2$.
    Let us obtain a branch-decomposition $(T_{i+1}',L_{i+1}')$ of $f_{i+1}$ 
    by attaching a new leaf corresponding to $v_{i+1}$ arbitrarily to $(T_i,L_i)$. 
    Then the width of $(T_{i+1}',L_{i+1}')$ is at most $2k$, 
    because for every $X\subseteq V_i$, $f_{i+1}(X)=f^*(X,V_{i+1}-X)\le f^*(X,V_i-X)+f^*(X,\{v_{i+1}\}) \le f_i(X)+f(\{v_{i+1}\})\le f_i(X)+k$. 
    By applying \zcref{prop:approx-exact} for~$f_{i+1}$ with $\ell=2k$, 
    we obtain a branch-decomposition $(T_{i+1},L_{i+1})$ of~$f_{i+1}$ of width at most $k$ if it exists 
    in time 
    $O(3^{8k}k n \SUB(n,kn,\gamma')\log n
    + \gamma' k 3^{(2k+1)(8k+9)}n)$.

    Thus the total running time to obtain a branch-decomposition of $f=f_n$ of width at most $k$, if it exists,
    is 
    $O(n\cdot (3^{8k}k n \SUB(n,kn,\gamma')\log n
    +  k 3^{(2k+1)(8k+9)}\gamma' n))$.

    If $f^*=f_{\min}$, then $\gamma'\le \SUB(n,kn,\gamma)$, because $f(X)\le\sum_{v\in X}f(\{v\})\le  k\abs{X}$ by the submodularity.
    Furthermore, minimizing $f_{i}$ can be done by computing $f_{\min}$, and therefore we can take $\SUB(n,kn,\gamma')=\SUB(n,kn,\gamma)$.
\end{proof}

\zcref{thm:main} almost implies \zcref{thm:main-simplified}, with the algorithm running only slightly slower than desired.

Instead of using the iterative compression, we can use the existing fixed-parameter tractable approximation algorithm for branch-width, due to Oum and Seymour~\cite{OS2004}. 
Hicks and Oum~\cite{HO2011} noted that the following theorem can be deduced from the argument of Oum and Seymour~\cite{OS2004}, which was originally written only for the case that $f(\{v\})\le 1$ for all $v\in V$.
\begin{theorem}[note={Oum and Seymour~\cite{OS2004}},label=thm:approx]
    Let $n$ be a positive integer.
    Let $f:2^V\to \mathbb Z$ be a connectivity function on an $n$-element set~$V$
    and let $c'=\max\{f(\{v\}):v\in V\}$. 
    Let $\gamma$ be the time to compute $f(X)$ for any set~$X$.
    Then in time $O(2^{3k+c'} n^2 \SUB(n,kn,\gamma))$, we can either find a branch-decomposition of $f$ of width at most $3k+c'$ 
    or confirm that the branch-width of~$f$ is larger than $k$.
\end{theorem}
By combining with \zcref{prop:approx-exact}, we deduce the following.
\begin{theorem}[label=thm:main2]
    Let $n$ be a positive integer.
    Let $f:2^V\to \mathbb Z$ be a connectivity function on an $n$-element set~$V$.
    Let us assume that $\gamma$ is the time to compute $f(X)$ for any subset $X$ of~$V$.
    In time 
$O(2^{4k} n^2 \SUB(n,kn,\gamma)+3^{16k}k n \SUB(n,kn,\gamma)\log n
    +  k 3^{(4k+1)(8k+9)}\gamma n)$, 
    we can either find a branch-decomposition of~$f$ of width at most~$k$, 
    or confirm that the branch-width of~$f$ is larger than~$k$.   
    
    In particular,
    the running time is at most 
    $2^{O(k)} \gamma n^6 \log n + 2^{O(k^2)}\gamma n$.
\end{theorem}
Note that \zcref{thm:main2} implies \zcref{thm:main-simplified}.
\begin{proof}
    We may assume that $f(\{v\})\le k$ for all $v\in V$, because otherwise the branch-width of~$f$ is larger than $k$.
    By \zcref{thm:approx}, in time $O(2^{4k} n^2 \SUB(n,kn,\gamma))$, we can find a branch-decomposition of width at most $4k$, unless we get a confirmation that the branch-width of~$f$ is larger than $k$.
    We now apply \zcref{prop:approx-exact} with $\ell=4k$.
    By \zcref{thm:submodularmin}, the running time is at most 
    $O( 2^{4k}n^2 k^3 \gamma n^4 \log (kn) + 3^{16k}  kn k^3 \gamma n^4\log (kn) \log n  
    +  k 3^{(4k+1)(8k+9)}\gamma n)    $.
\end{proof}

\section{Application to matroids given by the rank oracle}
\label{sec:matroid}
In this section we show that our algorithm can be implemented slightly faster for the case of matroid branch-width.
This is by using matroid intersection instead of general submodular minimization.

Note that we chose to use $\lambda(X)=r(X)+r(E(M)-X)-r(E(M))$ as the definition of the connectivity function of a matroid.
Some references include an additional $+1$ in the definition of $\lambda(X)$ and that will add $1$ to the branch-width.
\begin{theorem}[note={Matroid Intersection Theorem, Edmonds~\cite{Edmonds1970}},label=thm:edmonds]
    Let $M_1$, $M_2$ be two matroids on $E$ with rank functions $r_1$, $r_2$, respectively. 
    Then the maximum size of a common independent set is equal to 
    $\min_{U\subseteq E} (r_1(U)+r_2(E-U))$.
\end{theorem}

We use the following matroid intersection algorithm.

\begin{theorem}[note={Chakrabarty, Lee, Sidford, Singla, and Wong~\cite{CLSSW2019}},label=thm:matroidint]
    Let $M_1$, $M_2$ be two matroids on an $n$-element set given by rank oracles, and let $\gamma$ be the time to compute the rank of any set.
    Then we can find a largest common independent set in time $O(\gamma n \sqrt{r}\log n)$, where $r$ is the size of a largest common independent set.
\end{theorem}
Using the remark in \cite[page 526]{OS2004} combined with a newer algorithm for the matroid intersection problem stated in \zcref{thm:matroidint}, we deduce the following.
For completeness, we include the proof.
\begin{proposition}[note={Oum and Seymour~\cite[page 526]{OS2004}},label=prop:matroidmin]
There is an algorithm that, with input an $n$-element matroid $M$ 
    given by its rank oracle and two disjoint sets $X$ and $Y$, finds a 
    set $Z$ with $X\subseteq Z\subseteq E(M)-Y$ that minimizes $\lambda(Z)$ in time $O(\gamma n^{1.5}\log n)$, where $\lambda(Z)=r(Z)+r(E(M)-Z)-r(E(M))$ is the connectivity function of $M$  and $\gamma$ is the time to compute the rank of any set.     
\end{proposition}
\begin{proof}
    Let $r$ be the rank function of $M$.
    Let $M_1=M/X\setminus Y$ and $M_2=M\setminus X/Y$. 
    Let $r_1$, $r_2$ be the rank function of $M_1$ and $M_2$, respectively. 
    Then, for $U\subseteq E(M)-(X\cup Y)$, we have 
    \[ r_1(U)=r(U\cup X)-r(X) \text{ and } r_2(U)=r(U\cup Y)-r(Y).\] 
    The matroid intersection theorem, \zcref{thm:edmonds}, allows us to find $U$ that minimizes
    \begin{align*} 
        r_1(U)+r_2(E(M_2)-U) &= r(U\cup X)-r(X) + r(E(M)-(X\cup U)) - r(Y)\\ 
        &= \lambda(U\cup X)+ r(E(M))-r(X)-r(Y) .
    \end{align*} 
    By \zcref{thm:matroidint}, such $U$ can be found in time $O(\gamma n^{1.5}\log n)$.
    (Their paper is stated for finding a maximum common independent set, but it is easy to output the dual optimum, the set $U$ that minimizes $r_1(U)+r_2(E(M_2)-U)$ from the algorithm.)
\end{proof}
\begin{corollary}[note={Oum and Seymour~\cite[Corollary 7.2]{OS2004}},label=cor:matroid]
    For an integer $k$, there is an algorithm that, with input an $n$-element matroid, given by its rank oracle, either concludes that the branch-width is larger than $k$
    or outputs its branch-decomposition of width at most $3k+1$, in time $O(\gamma 8^k n^{2.5}\log n)$, where $\gamma$ is the time to compute the rank of any set.     
\end{corollary}
\begin{proof}
    The proof is identical to the proof in \cite{OS2004}, except for the fact that we now use an improved running time from \zcref{thm:matroidint,prop:matroidmin}.
\end{proof}

\getkeytheorem{matroidmain}
\begin{proof}
    Computing a branch-decomposition of width at most $3k+1$ takes the time $O(\gamma 8^k n^{2.5}\log n)$ by \zcref{cor:matroid}.
    Let $\ell=3k+1$. 
    We use \zcref{prop:approx-exact}
    but instead of using the general submodular function minimization in \zcref{prop:titanic}, we 
    will use the matroid intersection algorithm in \zcref{prop:matroidmin}.
    Thus $\SUB(n,kn,\gamma)$ can be replaced with $O(\gamma n^{1.5}\log n)$.
    So, by \zcref{prop:approx-exact}, the total running time is  at most 
    $O(8^k \gamma  n^{2.5}\log n+ 3^{12k} (3k+1) n (\gamma n^{1.5}\log n  ) \log n + k 3^{(3k+2)(8k+9)}\gamma  n)$.
\end{proof}

\section{Concluding Remarks}
\label{sec:conclusion}
Our algorithm manages to split the input instance into smaller instances in polynomial time whenever $n > 3^{4k+1}$, and otherwise calls the Oum-Seymour algorithm~\cite{OS2005}.
Therefore, it is natural to ask whether one can beat the Oum-Seymour algorithm in the setting when $n \le 2^{O(k)}$.
It would be interesting to obtain unconditional lower bounds for the number of oracle queries required.

Another interesting research direction would be to design optimized versions of our algorithm for important special cases such as graph branch-width and rank-width.

For a connectivity function $f:2^V\to\mathbb Z$, 
the \emph{path-width} of $f$ can be defined as the minimum~$k$ such that there is an ordering $v_1,v_2,\ldots,v_n$ of $V$ such that $f(\{v_1,v_2,\ldots,v_i\})\le k$ for all~$i$.
Depending on which connectivity functions we use, we obtain definitions of 
linear-width of graphs, cut-width of graphs, path-width of matroids, and linear rank-width of graphs.
For each of these width parameters, there are known fixed-parameter tractable algorithms to decide whether its value is at most $k$, \cite{TSB2005,TSB2000,Thilikos2000,BFT2009,Hamm2019,JKO2016}.
However, for the path-width of a general connectivity function, it is open whether it is fixed-parameter tractable to decide whether the path-width is at most~$k$. 
\begin{problem}
    Is there an algorithm to decide whether a connectivity function $f:2^V\to\mathbb Z$ given by an oracle has path-width at most $k$ in time $O(g(k) n^{c})$ for some function $g$ and constant $c$?
\end{problem}
The current best algorithm is due to Nagamochi~\cite{Nagamochi2012}. 
He proved that there is a $\gamma n^{O(k)}$-time algorithm to decide whether the path-width of a connectivity function~$f$ is at most $k$, where $\gamma$ is the time to compute $f(X)$ for any set $X$.

\bibliographystyle{amsplain}
\providecommand{\bysame}{\leavevmode\hbox to3em{\hrulefill}\thinspace}
\providecommand{\MR}{\relax\ifhmode\unskip\space\fi MR }
\providecommand{\MRhref}[2]{\href{http://www.ams.org/mathscinet-getitem?mr=#1}{#2}
}
\providecommand{\href}[2]{#2}

\end{document}